\documentclass[a4paper,11pt]{article}
\usepackage[margin=1in]{geometry}

\usepackage{hyperref}
\hypersetup{colorlinks=true,allcolors=blue}
\usepackage{url}
\usepackage{amsmath,amsthm,amssymb,amsfonts}
\usepackage{xspace}
\usepackage[dvipsnames]{xcolor}
\usepackage[linesnumbered,ruled,vlined]{algorithm2e}
\usepackage[page]{appendix}
\usepackage{thmtools}
\usepackage{cleveref}

\usepackage{multirow}
\usepackage{booktabs}
\theoremstyle{plain}
\newtheorem{theorem}{Theorem}[section]
\newtheorem{lemma}[theorem]{Lemma}

\newtheorem*{claim*}{Claim}

\newtheorem{fact}[theorem]{Fact}

\crefname{lemma}{lemma}{lemmas}
\Crefname{lemma}{Lemma}{Lemmas}
\Crefname{claim}{Claim}{Claims}
\Crefname{fact}{Fact}{Facts}

\theoremstyle{definition}
\newtheorem{definition}[theorem]{Definition}

\theoremstyle{remark}

\title{Streaming Max-Cut in General Metrics} 

\author{Shaofeng H.-C. Jiang\thanks{Peking University. Email: \texttt{shaofeng.jiang@pku.edu.cn}}
\and
Pan Peng\thanks{University of Science and Technology of China. Email: \texttt{ppeng@ustc.edu.cn}}
\and
Haoze Wang\thanks{Peking University. Email: \texttt{2200012915@stu.pku.edu.cn}}
}

\DeclareMathOperator{\cut}{cut}
\DeclareMathOperator{\dist}{dist}
\DeclareMathOperator{\poly}{poly}
\DeclareMathOperator{\med}{med}

\newcommand{\maxcut}{\ensuremath{\operatorname{Max-Cut}}\xspace}
\newcommand{\E}{\mathbb{E}}

\newcommand{\X}{\mathcal{X}}
\newcommand{\Y}{\mathcal{Y}}

\begin{document}

\maketitle
   \begin{abstract}
    \maxcut is a fundamental combinatorial optimization problem that has been studied in various computational settings.
    We initiate the study of its streaming complexity in \emph{general metric spaces} with access to distance oracles. 
    We give a $(1 + \epsilon)$-approximate algorithm for estimating the \maxcut value in \emph{sliding-window} streams using only poly-logarithmic space. 
    This is the first sliding-window algorithm for \maxcut even in Euclidean spaces,
    and it matches a known insertion-only space bound in the special case of Euclidean spaces [Chen, Jiang, Krauthgamer, STOC'23].
    In sharp contrast, we give a $\poly(n)$-space lower bound in the \emph{dynamic} streaming setting. This yields a separation from the Euclidean case, where the polylogarithmic-space $(1+\epsilon)$-approximation extends to dynamic streams.
    
On the technical side, our sliding-window algorithm builds on the smooth histogram framework of [Braverman and Ostrovsky, SICOMP'10]. To make this framework applicable, we establish the first smoothness bound for metric \maxcut. Moreover, we develop a streaming algorithm for metric \maxcut in insertion-only streams, whose key ingredient is a new metric reservoir sampling technique.

\end{abstract}

\section{Introduction}

\maxcut is a fundamental combinatorial optimization problem and it has been a central topic in algorithm study
under various settings, including polynomial-time approximation algorithms~\cite{GoemansW95,fernandez1996max,mathieu2008yet,DBLP:journals/corr/Yaroslavtsev14}, streaming algorithms~\cite{frahling2005coresets,DBLP:conf/stoc/ChenJK23,MenandW26}, and sublinear-time algorithms~\cite{goldreich1998property,alon2003random,rudelson2007sampling,indyk2001high,peng2023sublinear}.
In this paper, we initiate the study of the streaming complexity of metric \maxcut.
In metric \maxcut,
we are given an (underlying) general metric space $(V, \dist)$ and a set of input data points $P \subseteq V$.
The goal is to find a subset $S\subseteq P$ such that its cut value, defined as
\begin{equation}
    \cut_P(S) := \sum_{x \in S} \sum_{y \in P \setminus S} \dist(x, y)
\end{equation}
is maximized.
In the streaming setting, the data set $P$ with $n$ points is presented as a data stream of point IDs,
and the distances between two points can be accessed via a distance oracle:
the algorithm can learn the distance between any two points
provided that the corresponding IDs are currently stored by the algorithm, without additional space consumption.
The objective is to output an estimate of the optimal \maxcut \emph{value} of the underlying graph defined by the stream, while using as little space as possible. 
The focus is on estimating the optimal \maxcut \emph{value}, instead of the solution $S \subseteq V$ since storing the solution may already take $\Omega(n)$ space and this trivializes the problem.

This streaming metric \maxcut problem has been recently studied in the special case when the metric is Euclidean,
and $(1 + \epsilon)$-approximation has been obtained using small space~\cite{frahling2005coresets,DBLP:conf/stoc/ChenJK23,MenandW26}, even when the stream has both insertion and deletions.
However, these results crucially exploit the Euclidean structure. A fundamental question is whether such structural assumptions are necessary for obtaining $(1+\epsilon)$-approximations in the streaming model, and whether the results can be extended to arbitrary metric spaces.

\subsection{Our Results}
In this work, we systematically address the streaming complexity of metric \maxcut, and this is the first set of results for streaming \maxcut in general metrics. We focus on three standard models: (1) \emph{Insertion-only streams}: the stream consists solely of point insertions.
(2) \emph{Sliding-window streams}: the stream consists of point insertions, but only the most recent $w$ points ($w\geq 1$) are considered active, and the algorithm must maintain an estimate after every update. When $w = \infty$, this reduces to the insertion-only case. (3) \emph{Dynamic streams}: the stream allows both insertions and deletions of points.
For both the insertion-only and dynamic streaming setting, the algorithm needs only output at the end of the stream.

Our key conceptual message is that the streaming complexity of \maxcut 
demonstrates a strong separation between the insertion-only streams (and more generally sliding-window streams) and dynamic streams,
by giving competitive streaming algorithms for sliding-window streams and a strong lower bound for dynamic streams. This also separates the Euclidean case from the general metric case, since in Euclidean spaces, $(1 + \epsilon)$-approximation with poly-logarithmic space is known even for dynamic streams~\cite{frahling2005coresets,DBLP:conf/stoc/ChenJK23,MenandW26}.

As in the standard litereature of streaming algorithms,
we measure the space in words, 
and each word can store $\poly\log N$ bits, where $N$ is the input size.
We condiser the following distance oracle model throughout.
\begin{definition}[Distance Oracle Model]
    \label{def:model}
Each metric point is represented by a unique identifier (ID). The algorithm accesses
distances only through a distance oracle, which takes two point IDs and
returns their metric distance. The oracle representation itself is not charged
to the algorithm's space complexity,
and we assume each ID or distance fits in one word.
\end{definition}

\paragraph{Sliding-Window Algorithms.}

Our first result is a $(1 + \epsilon)$-approximation to the metric \maxcut value using poly-logarithmic space in sliding-window streams.

\begin{theorem}[Informal; see \Cref{thm:sliding_window}]
\label{thm:intro_sliding_window}
    There is sliding-window algorithm for metric \maxcut that estimates the optimal value for each window within $(1 + \epsilon)$ factor with constant probability under the distance oracle model,
    using space $\poly(\epsilon^{-1}\log(w \Delta))$ where $\Delta$ is the aspect ratio of the metric space and $w$ is the window size.
\end{theorem}
The above result is based on a new algorithm for metric \maxcut in the insertion-only streams with polylogarithmic space (see \Cref{thm:insertion}). As mentioned, prior streaming bounds for metric \maxcut were known only in Euclidean spaces,
and the state-of-the-art~\cite{DBLP:conf/stoc/ChenJK23,MenandW26} was a $(1 + \epsilon)$-approximation using space $\poly(\epsilon^{-1} d \log \Delta)$ for points coming from $[\Delta]^d$ ($\Delta \in \mathbb{N}$ and $[\Delta] := \{1, \ldots, \Delta\}$). Furthermore, it is unknown if this Euclidean bound holds for sliding window\footnote{While the result of~\cite{DBLP:conf/stoc/ChenJK23} works for dynamic streams, their setting is not comparable to our sliding-window streams
as the deletions must be explicitly given in their dynamic streaming setting whereas the sliding-window deletes an element implicitly.
}.
Our result nearly matches this competitive trade-off,
works for sliding window, and works for general metrics which does not rely on any specific structure of Euclidean spaces\footnote{Our result readily applies to Euclidean case with an additional $O(d)$ factor in space complexity, since the trivial distance oracle that returns the $\ell_2$ distance between two points in $\mathbb{R}^d$ only uses $O(d)$ space.
}.

\paragraph{Lower Bound for Dynamic Streams.}
Our second result is a lower bound for dynamic streams, and it shows that no algorithm using small space can achieve even $\poly(n)$-approximation in general metric spaces.
Here, the input of a dynamic stream consists of both insertions and deletions of point (IDs), and the algorithm needs to report the (estimate of the) \maxcut value at the end of the stream,
with respect to the inserted points that are not yet deleted.
The distance oracle model is the same as in \Cref{def:model},
where the algorithm must provide the IDs of the point pair for querying the distance.

\begin{theorem}[Informal; see \Cref{thm:lb}]
\label{thm:intro_lb}
There is no dynamic streaming algorithm that estimates the metric \maxcut value within $O(\Delta / n^{1 / 3})$ factor using space $o(n^{1/3})$ with constant probability, where $\Delta$ is the aspect ratio and $n$ is the size of the stream.
\end{theorem}

This clearly shows that the dynamic streaming setting is significantly ``harder'' than the sliding-window setting or insertion-only setting for metric \maxcut.
In addition, this also shows an interesting difference to the Euclidean case:
In Euclidean spaces, $(1 + \epsilon)$-approximation using poly-logarithmic space (similar to that in \Cref{thm:intro_sliding_window}) work also for dynamic streams~\cite{frahling2005coresets,DBLP:conf/stoc/ChenJK23}.

We remark that similar model for our lower bounds, which reveals the distance only upon query, were also used in e.g.~\cite{Cohen-AddadSS16};
in fact, the model in~\cite{Cohen-AddadSS16} is even more restricted,
where the algorithm is assumed to store the IDs explicitly,
and the space complexity is counted as the number of stored IDs.
However, a recent work~\cite{KhannaPSW25} introduces
a stronger notion: the metric is public such that the algorithm is allowed to query any pairs, not only for the IDs in the stream (which is a subset of the full metric points).
Under this model, a lower bound for diameter problem is established~\cite{KhannaPSW25}, and it is an interesting open question to establish a \maxcut lower bound under the same model (which will be stronger than \Cref{thm:intro_lb}).

\vspace{-0.5em}

\subsection{Technical Overview}
\label{sec:tech_overview}
\subsubsection{Upper Bounds}
At a high level, 
we make use of the smooth histogram framework~\cite{braverman2010effective} to obtain the sliding-window algorithm.
This framework requires us to a) devise an insertion-only streaming algorithm for metric \maxcut, and b) to derive a smoothness upper bound for metric \maxcut.
These two ingredients are both nontrivial for general metric streams,
and are our main technical contributions.

\paragraph{An Insertion-only Algorithm.}
For insertion-only streams, our high-level idea is to (approximately) implement an importance sampling framework in general metrics, 
which is suggested in~\cite{DBLP:conf/stoc/ChenJK23} to devise a Euclidean streaming algorithm for \maxcut.
In this framework, each data point $x \in P$ is sampled with probability roughly proportional to its total distance from all other points, denoted $q(x):=\sum_{y \in P} \mathrm{dist}(x, y)$.
With a sample $S \subseteq P$ of size $\poly(\epsilon^{-1})$, we then evaluate $\maxcut(S)$.
As in~\cite{DBLP:conf/stoc/ChenJK23} (restated in \Cref{lem:sampling lower bound}), this evaluation yields a $(1+\epsilon)$-approximation to $\maxcut(P)$. 
In our distance oracle model, evaluating $\maxcut(S)$ is straightforward since we can query the oracle to obtain all pairwise distances within $S$.  

The main challenge lies in performing the importance sampling.
In particular, $q(x)$ may rely on the entire stream, and it is not easy to estimate upon the arrival of $x$.
To resolve this issue, our key observation is that, an approximate $\hat q(x)$ which estimates the sum of distances from $x$ to only the prefix until the arrival of $x$,
is already a good approximation for $q(x)$,
and our algorithm may be viewed as an importance sampling with respect to $\hat q(x)$'s.

Specifically, we propose the following modified reservoir-sampling procedure.
Reservoir-sampling is a classic approach is to generate uniform samples from an insertion-only stream,
and here we augment it to return importance sampling.
The process maintains a single sample $s \in P$.
Label the data points in $P := \{p_1, \ldots, p_n\}$ according to their order of arrival in the stream, and let $P_i := \{p_1, \ldots, p_i\}$ denote the prefix of the first $i$ points.  
When the $i$-th point $p_i$ is inserted
we consider the prefix-based quantity $q'(p_i) := \sum_{x \in P_i} \dist(p_i, x)$.
Then we replace the current sample with $p_i$ with probability $\frac{\hat q(p_i)}{\sum_{j \leq i} \hat q(p_j)}$,
and do nothing otherwise.

However, it is nontrivial to maintain $\hat q(x)$ precisely, and we use a coreset method to approximate $\hat q(x)$'s.
Observe that the definition of $\hat q(x)$ coincides with the $1$-median objective $\med(P_i, x)$: for $U \subseteq V$ and $y \in V$, $\med(U, y) := \sum_{y' \in U} \dist(y', y)$. 
It is known that an $\epsilon$-coreset (see \Cref{def:coreset}) for $1$-median can be maintained efficiently in insertion-only streams (see \Cref{lem:coreset}).  
This result serves as a key building block for our final algorithm, and the $\epsilon$-approximation error is carefully accounted for in the analysis of \Cref{lem:sampling lower bound}.

\paragraph{Smoothness of Metric \maxcut.}
To leverage the framework of \cite{braverman2010effective}, we need to show that our insertion-only algorithm provides an approximation of a smooth function. 
Roughly speaking, a function $f$ is smooth, if for any $0 < \epsilon < 1$,
there exist $0 < \beta < \alpha < 1$,
such that if some $B \subseteq A \subseteq V$ satisfies $(1 - \beta) f(A) \leq f(B)$
then $(1 - \alpha) f(A \cup C) \leq f(B \cup C)$ for every $C \subseteq V$.
A natural attempt is to prove the smoothness for $f = \maxcut$.
However, somewhat surprisingly, $f(\cdot) = \maxcut(\cdot)$ does \emph{not} satisfy this smoothness definition.
To see this, consider the 1D case, and $f$ does not satisfy the condition for 
$A =\{1, \ldots, 1\}$, then $B = \{1\}$, $C = \{0\}$,
because $f(A) = f(B) = 0$ but $f(A \cup C) =n > f(B \cup C) / (1 - \alpha) = 1 / (1 - \alpha)$ for any $\alpha$, where $n$ is a sufficiently large number.

We identify that the issue is caused by $0$-valued $f$ (e.g., $f(A) = f(B) = 0$ as in the above example).
Hence, our new idea is to add a tiny term to $\maxcut$, and define $f(S) := \maxcut(S) + \frac{\epsilon}{n} |S|$ to bypass this issue,
assuming the minimum non-zero distance is $1$ in the metric (otherwise one can rescale).
We show that this new $f$ is indeed smooth.
On the other hand, this tiny term does not add much error to $\maxcut$, so estimating $f$ suffices for estimating $\maxcut$.
Indeed, if $\maxcut(P) > 0$ (where $P$ is the dataset), then the added term is at most $\epsilon \maxcut(P)$,
due to an assumption that the minimum non-zero distance is $1$. Hence, in this case, we can readily use the smooth histogram framework and reduce to estimating $f$ in an insertion-only stream.
Finally, the other case of $\maxcut(P) = 0$ can be easily detected in a sliding window and hence can be handled separately.

\subsubsection{Lower Bounds}
We prove the lower bound directly instead of reducing to some known hard problems.
We construct two hard input instances that differ only locally and slightly, but their \maxcut values differ significantly. We then show that it is hard to differentiate the two inputs.

At a high level, the metric of the hard input consists of many clusters such that the distance between the clusters is $\Delta\geq \Omega(n^{1/3})$ (where $n$ is the number of points),
and inside the clusters the pairwise distance is $1$ (very small compared with $\Delta$),
except for one special cluster of size $\Theta(n^{1/3})$, denoting whose point set as $S$.
This special cluster $S$ has an outlier point whose distance to every other $|S| - 1$ points in $S$ is some parameter $k \in \{1, \Delta\}$, thereby defining the two hard instances.
The stream first sequentially inserts all points, then delete every point except for the special cluster $S$, so this special cluster $S$ is the final input dataset.
If $k = 1$ then $\maxcut(S) \leq O(|S|^2)< \Delta|S|$,
whereas if $k = \Delta$ then $\maxcut(S) \geq \Delta|S|$. 
Furthermore, any algorithm with approximation ratio $o({\Delta}/{|S|})$ must be able to tell if $k = 1$ or $k = \Delta$, and this is the main claim of our proof. 

Without diving into the proof details, we discuss the intuition why our input is hard in a conceptual sense.
The first angle is how the deletions make it hard.
Indeed, the \maxcut for the entire input is easy to approximate as ignoring the special cluster, or even ignoring any constant fraction of clusters, does not change the objective by much.
However, the deletion forces the algorithm to learn the structure of a tiny special cluster $S$ accurately,  which is a hard task.
Another angle is how we utilize the general metric and distance oracle setting, especially how this is different from the Euclidean case.
In fact, the metric in our hard input can be embedded into Euclidean spaces without large distortion, hence it is not the metric itself that is hard. 
What makes it harder than Euclidean case,
is that the algorithm can only learn the distances through a distance oracle, without the access to the vector representation in Euclidean spaces.
In particular, this forces the algorithm to explicitly query the distance between the outlier point and some other point in $S$ to differentiate whether $k = 1$ or $k = \Delta$ (which is crucial for approximating the \maxcut),
and explicitly discover the outlier point is difficult. 
On the other hand, in Euclidean spaces one can use structures such as (space efficient) tree embedding to approximately learn the cluster structure and differentiate the outlier, without explicitly knowing the outlier point (see e.g.,~\cite{DBLP:conf/stoc/ChenJK23}).

Due to space constraint, we defer the proof of the lower bound to \Cref{app:lowerbound}.

\vspace{-0.5em}

\subsection{Related Work}

The approximation algorithms for \maxcut has been very well studied.
In general graphs, the seminal result of Goemans and Williamson~\cite{GoemansW95} gave a $0.878$-approximation,
and this ratio is shown to be tight assuming the Unique Games Conjecture (UGC)~\cite{DBLP:journals/siamcomp/KhotKMO07}.
If the input is restricted to dense graphs (roughly, graphs with $\Omega(n^2)$ edges), 
then \maxcut admits PTAS's~\cite{GoemansW95,fernandez1996max,mathieu2008yet,DBLP:journals/corr/Yaroslavtsev14},
and the state-of-the-art achieves near-linear running time~\cite{DBLP:journals/corr/Yaroslavtsev14}.
Similar PTAS's also exist for general metric spaces (which in a sense is a dense graph)~\cite{DBLP:journals/jcss/VegaK01, indyk2001high}.

Under streaming model,
for graph streams, it is shown that breaking the $2$-approximation barrier requires $\Omega(n)$ space, even in insertion-only streams and when the edges are present in a random order~\cite{DBLP:conf/stoc/KapralovK19}. 
In contrast, for geometric streams in $\mathbb{R}^d$, 
it is possible to obtain $(1 + \epsilon)$-approximation in polylogarithmic  space~\cite{DBLP:conf/stoc/ChenJK23}.
Recently, this result has been extended to provide not only an approximate value but also an oracle to return an approximate \maxcut solution~\cite{MenandW26}. Furthermore, Dong et al. \cite{dong2025learning} recently gave a learning-augmented streaming algorithm that approximates the value of \maxcut of a general graph with an approximation ratio slightly better than $2$.  \vspace{-0.5em}

\section{Preliminaries}
\label{sec:prelim}
For $m\in \mathbb{N}$, write $[m]:=\{1,2,\ldots,m\}$.
Let $(V, \dist)$ by the underlying metric space. Given a point set $P$, Given a point set $P$, the \emph{aspect ratio} of $P$ is defined as the ratio between its largest and smallest non-zero pairwise distances.  
For some $\alpha \geq 1$, 
we say a non-negative real number $E$ is $\alpha$-approximate to $\maxcut(P)$ for dataset $P$,
if $E \leq \maxcut(P) \leq \alpha E$.
We interpret a set as a multi-set throughout.
The cut function $\cut : V \times V \to \mathbb{R}_+$ is defined as
$\cut(S, T) := \sum_{x \in S} \sum_{y \in T} \dist(x, y)$ for $S, T \subseteq V$.
For a subset $S\subseteq V$ and a dataset $P \subseteq V$,
write $\cut_P(S) := \cut(S, P \setminus S)$.
The \maxcut value of a dataset $P \subseteq V$ is defined as
\[
    \maxcut(P) := \max_{S \subseteq P} \cut_P(S).
\]
For $\alpha \geq 1$, an $\alpha$-approximate solution for Max-Cut of $P$ is 
a subset $S\subseteq V$ such that 
\[
\maxcut(P) / \alpha \leq \cut_P(S) \leq   \maxcut(P)
\]

\begin{definition}[Weighted Set]
    A weighted set is a pair $(S,w_S)$ where $S$ is a finite set of points and $w_S: S \to \mathbb{R}_{\geq 0}$ assigns a nonnegative weight to each element of $S$.
    \end{definition}
    \begin{definition}[\maxcut on Weighted Sets]
        Let $(S, w_S)$ be a weighted set.
        We define 
        \[\maxcut(S):=\max_{T\subseteq S}\sum_{x\in T}\sum_{y\in S\setminus T}w_S(x)w_S(y)\dist(x,y).\]
    \end{definition}

 \section{Algorithms for Insertion-Only Streams}
\label{sec:insertion_only}

In this section, we present an algorithm to approximate the value of \maxcut for insertion-only streams and establish the following theorem.
The sliding-window result is based on this insertion-only one,
and will be presented in~\Cref{sec:sliding_window}.

\begin{theorem}
    \label{thm:insertion}
    Suppose $(V, \dist)$ is an underlying metric space.
    There is an algorithm that given as input
    $\epsilon \in (0, 1/2)$, $\Delta \geq 1$,
    and a point set $P \subseteq V$ presented as an insertion-only stream, 
    such that if $\dist(x, y) > 0$ then $1\le \dist(x,y)\le \Delta$ for all $x,y\in P$,
    outputs a $(1 + \epsilon)$-approximation to $\maxcut(P)$
with probability at least $2/3$,
using space $\poly(\frac{\log (n\Delta)}{\epsilon})$ where $n$ is the size of $P$.
\end{theorem}

For technical reasons, we assume throughout that $\maxcut(P) > 0$; the case $\maxcut(P) = 0$ is trivial. In the following, we first present the necessary tools; then we describe the algorithm; we analyze it and prove \Cref{thm:insertion}.
\subsection{Some Tools} 
Our high-level framework is the following importance sampling
suggested by~\cite{DBLP:conf/stoc/ChenJK23}. 
Roughly speaking, one samples each data point $x$ with probability (approximately) proportional to the sum of distances from $x$ to every other points,
then around $\poly\log n$ samples is enough to approximate the \maxcut value.

\begin{lemma}[Importance Sampling~\cite{DBLP:conf/stoc/ChenJK23}]
    \label{lem:chen}
    Given $\varepsilon, \delta > 0$, $\lambda \geq 1$, 
    metric space $(V, \dist)$ and dataset $P \subseteq V$, 
    let $\mathcal{D}$ be a distribution $(p_x : x \in P)$ on $P$ such that $\forall x \in P,\, p_x \geq \frac{1}{\lambda} \cdot \frac{q(x)}{Q}$, where $q(x) := \sum_{y \in X} \dist(x, y)$
    and $Q := \sum_{x \in X} q(x)$.
    Let $S$ be a \emph{weighted} set that is obtained by an \textit{i.i.d.} sample of $m \geq 2$ points from $\mathcal{D}$, weighted by $w_S(x) := \hat{p}_x$ such that $p_x \leq \hat{p}_x \leq (1 + \varepsilon) \cdot p_x$. 
    If $m \geq O(\varepsilon^{-4} \lambda^8)$, then with probability at least $0.9$, the value
    $\frac{\mathrm{Max\text{-}Cut}(S)}{m^2}$
    is a $(1 + \varepsilon)$-approximation to $\mathrm{Max\text{-}Cut}(P)$.
\end{lemma}

We also require the following coreset tool for $1$-median clustering.  
Given a set $U \subseteq V$ and a point $y \in V$, the $1$-median cost of $y$ with respect to $U$ is defined as $\sum_{y' \in U} \dist(y', y)$.
\begin{definition}[{Coreset~\cite{Har-PeledM04}}]
    \label{def:coreset}
    A weighted set $S \subseteq P$ with weight $w_S' : S \to \mathbb{R}_+$
    is called an $\epsilon$-coreset for $1$-median,
    if
\[
    \forall x \in P, \qquad
     \sum_{p \in P} \dist(p, x) \in
     (1 \pm \epsilon) \cdot \sum_{s \in S} \dist(x, s) \cdot w_S'(s).
\]
\end{definition}

\begin{lemma}[Streaming Algorithms for Coreset\cite{Chen09}]
    \label{lem:coreset}
There exists a streaming algorithm that, given $\epsilon \in (0,1)$ and a set $P \subseteq V$ with aspect ratio $\Delta$ presented as an insertion-only stream $p_1,p_2,\dots,p_n$, maintains an $\epsilon$-coreset $(C,w_C’)$ for the 1-median problem on $P$. Specifically, with probability $1-\frac{1}{n^2}$, for every timestamp $t \geq 1$, letting $P_t = {p_1,\ldots,p_t}$ denote the first $t$ points, the coreset satisfies: \[
\forall x \in P_t, \qquad
\sum_{p \in P_t} \dist(p,x) \in (1 \pm \epsilon) \cdot \sum_{s \in C} \dist(x,c)\,w_C'(c).
\]
The coreset uses $\poly(\epsilon^{-1} \log (n\Delta))$ space.
\end{lemma}

\subsection{The Algorithm}
We propose a variant of importance sampling as in~\Cref{lem:chen}.
Our importance sampling uses some $q'$ as the importance score, 
which is defined with respect to only a prefix of the stream, instead of $q$.  
Specifically, given a data stream $p_1, \dots, p_n$ of points in $P$, let $P_i := \{p_1, \ldots, p_i\}$ be the prefix of the first $i$ points. 
For each $p_i$, we define $q'(p_i) := \sum_{x \in P_i} \dist(p_i, x)$ with respect to the prefix.
We maintain an $\epsilon$-coreset for the $1$-median objective $\med(P_i, x)$ (as in \Cref{lem:coreset}),
and use it to maintain a single sample $s \in P$:
when $p_i$ arrives, we replace the current sample with $p_i$ with probability (approximately) proportional to $\frac{q'(p_i)}{\sum_{j \leq i} q'(p_j)}$, 
and otherwise keep the current sample. We give the full algorithm in \Cref{algo:sample one}.

\begin{algorithm}[t]
    \caption{Sampling one point $s$ from $P$ represented as a stream $p_1, \dots, p_n$.}
    \label{algo:sample one}
    \DontPrintSemicolon
    $s \gets \textbf{null}$, $\hat{Q}_0 \gets 0, w \gets 1, K \gets \Theta(\log(\Delta n^2))$\;
    initialize $(C,w_{C}')$ as an $\epsilon$-coreset for $1$-median (initially  $C=\emptyset$, $w_{C}'(c)=0$ for any $c$)\;
\For{$t = 1, 2, \ldots$}{
        $\hat{R}_t \gets$ $\sum_{c \in C} \dist(c, p_t) w_{C}'(c)$\;
$\hat{Q}_t \gets \hat{Q}_{t-1} + 2\hat{R}_t$\;
        \If{$\hat{Q}_t = 0$}{
            $\beta_t \gets \frac{1}{t}$
        }
        \Else{
            $\beta_t \gets \frac{1}{K} \cdot \frac{\hat{R}_t}{\hat{Q}_t}$\;
        }
        with probability $\beta_t$, $s \gets p_t$; if $s$ is updated, $w \gets \beta_t$, otherwise $w \gets w \cdot (1 - \beta_t)$\;
        update the coreset $(C,w_{C}')$ as specified in \Cref{lem:coreset}\;
}
    \KwRet $(s, w)$\;
\end{algorithm}

The final algorithm for estimating $\maxcut(P)$ is given in \Cref{algo:maxcut},
where we plug in \Cref{lem:chen} the approximate estimate as in \Cref{algo:sample one}.
\begin{algorithm}[t]
    \caption{Estimating $\maxcut$ of the set $P$ represented as a stream $p_1, \dots, p_n$.}
    \label{algo:maxcut}
    \DontPrintSemicolon
    \If{$\dist(p_i,p_j)=0$ holds for all $i,j$}{
    \KwRet $0$\;
    }
    let $K$ be as in \Cref{algo:sample one}, and let $\lambda := \tfrac{4K(1+\epsilon)}{1-\epsilon}$\;

    maintain $m := \Theta(\epsilon^{-4}K^{8})$ independent runs of \Cref{algo:sample one} during the update of the stream\;

    query each run to obtain a sample $(s_i, w_i)$ to
    define $(S, w_S) := (\{s_i\}_{i=1}^m, \{w_i\}_{i=1}^m)$\; 
\KwRet $\eta:=\frac{\mathrm{Max\text{-}Cut}(S)}{m^2}$ as in \Cref{lem:chen}\;
\end{algorithm}

\subsection{Proof of \Cref{thm:insertion}}
We start with the analysis of \Cref{algo:sample one}. 
Define $R_t := \sum_{i=1}^t \dist(p_t, p_i)$ as the sum of distances from $p_t$
to the prefix $P_t$,
and $Q_t := \sum_{i=1}^t \sum_{j=1}^t \dist(p_i, p_j)$ be the total sum of all distances among the prefix $P_t$.
Without less of generality, we assume $Q_n>0$.
Recall that $s$ is the sample returned by \Cref{algo:sample one}.
We show that for any point $x \in P$, the probability that $s = x$ is at least $\Omega\left(\frac{1}{K} \cdot \tfrac{q(x)}{Q}\right)$, where $q(x) = \sum_{y \in P} \dist(x,y)$ and $Q = Q_n$.
This bound essentially says that the importance score $q'(x)$ that is defined only with respect to the prefix,
is a good approximation for the $q(x)$'s defined with respect to the full stream.
This further allows us to invoke \Cref{lem:chen} and thereby ensure the correctness of our algorithm.
Formally, we establish the following key lemma.

\begin{lemma}
    \label{lem:sampling lower bound}
Let $(s, w)$ denote the return value of \Cref{algo:sample one}.
With probability $ 1 - 1 / n^2$, for every $p_t \in P$,
the following inequality holds:
\[
    w = \Pr[s = p_t] \geq \frac{1-\epsilon}{4K(1+\epsilon)} \cdot \frac{\sum_{i=1}^n \dist(p_t, p_i)}{Q_n}.
\]
\end{lemma}

The proof of this lemma is postponed to \Cref{sec:proof_key_lemma},
and we first conclude the proof of \Cref{thm:insertion} assuming \Cref{lem:sampling lower bound} is correct.

\begin{proof}[Proof of \Cref{thm:insertion}]
We now show that \Cref{algo:maxcut} outputs a good estimate $\eta$ for $\maxcut(P)$. It suffices to consider the non-trivial case where there exist $p_i, p_j$ with $\dist(p_i,p_j)>0$. 

Note that $K = O(\lambda)$.
Our algorithm invokes \Cref{algo:sample one} $m$ times with $m = O(\epsilon^{-4}K^{8}) = O(\epsilon^{-4}\lambda^{8})$
to obtain the sample set $S$ and corresponding weights.
\Cref{lem:sampling lower bound} ensures that each sample $s$ with weight $w$ from \Cref{algo:sample one}
satisfies the premises of \Cref{lem:chen},
where the distribution $\mathcal{D}$ is given by $p_x = \Pr[s=x] \geq \tfrac{1}{\lambda}\cdot \tfrac{q(x)}{Q}$, with $q(x) := \sum_{y \in P} \dist(x,y)$ and $Q := \sum_{x \in P} q(x)$. Furthermore, the weight is $w_S(x) = \Pr[s=x]$. Since we invoke \Cref{algo:sample one} for $m \ll n$ times, \Cref{lem:chen} implies that with probability at least $1 - \tfrac{m}{n^2} - 0.1 > \tfrac{2}{3}$, the estimate $\eta$ is a $(1+\epsilon)$-approximation of $\maxcut(P)$.

 To bound the space complexity, note that by \Cref{lem:coreset}, a single invocation of \Cref{algo:sample one}
 requires $\poly(\epsilon^{-1}\log(n\Delta))$ space to maintain the coreset. Since we invoke \Cref{algo:sample one} $m$ times, the total space for all coresets is
 $m \cdot \poly(\epsilon^{-1}\log(n\Delta)) = \poly(\epsilon^{-1}\log(n\Delta))$.
 After obtaining the set $S$, we query all $O(m^2)$ pairwise distances to estimate the weighted max cut, which requires $O(m^2) = \poly(\epsilon^{-1}\log(n\Delta))$ space.
 Hence, the overall space complexity is $\poly(\epsilon^{-1}\log(n\Delta))$.
\end{proof}

\subsection{Proof of Key \Cref{lem:sampling lower bound}}
\label{sec:proof_key_lemma}

\begin{proof}
    Since we consider $s=p_t$, then $p_t$ is sampled when it comes,
    and all subsequent points are not sampled, so the sampling probability $\Pr[s = p_t]$ is given by:
    $w = \Pr[s = p_t] = \beta_t \cdot \prod_{i = t+1}^n (1 - \beta_i)$.

\paragraph{Case I: $Q_t > 0$.}
    In the following, we consider the case that $Q_t > 0$.
    Using the inequality $\prod_i(1 - x_i) \geq 1 - \sum_i x_i$ for $0\le x_i\le 1$, we have: 
    $\Pr[s = p_t] \geq \beta_t \cdot \left(1 - \sum_{i = t+1}^n \beta_i\right)$. 
Substituting the definition of $\beta_i$, this becomes:
    \begin{equation}
    \label{eqn:Pr_s_pt}    
    \Pr[s = p_t] \geq \beta_t \cdot \left(1 - \frac{1}{K} \cdot \sum_{i = t + 1}^n \frac{\hat{R}_i}{\hat{Q}_i}\right).
    \end{equation}
Next, we bound $\frac{\hat R_i}{\hat{Q}_i}$ in the following lemma.
\begin{lemma}
\label{obs:R}
With probability at least $1-\frac{1}{n^2}$, for every timestamp $t$, it holds that 
\begin{align}
(1-\epsilon) R_t  &\leq  \hat{R}_t  \leq  (1+\epsilon) R_t.\label{eqn:Rt} \\
    (1-\epsilon) Q_t &\leq \Hat{Q}_t \leq (1+\epsilon) Q_t.\label{eqn:Qt}
\end{align}
\end{lemma}
\begin{proof}
By \Cref{lem:coreset}, it holds that with probability at least $1-1/n^2$, for any $t$ and for all $x \in P_t$, 
$\sum_{p \in P_t} \dist(p,x) \in (1 \pm \epsilon) \cdot \sum_{s \in C} \dist(x,c)\,w_C'(c)$. In the following, we assume this event holds.

Note that $R_t = \sum_{i=1}^t \dist(p_t,p_i)$ and $\hat{R}_t = \sum_{c \in C} \dist(c,p_t)w_C'(c)$. Thus \eqref{eqn:Rt} holds for all $t$.
Furthermore, since $\Hat{Q}_t = \sum_{i=1}^t \Hat{R}_i$, applying \eqref{eqn:Rt} 
    we get
        $\Hat{Q}_t = \sum_{i=1}^t \Hat{R}_i \ge (1-\epsilon)\cdot \sum_{i=1}^t R_t = (1-\epsilon) Q_t$. 
Similarly, $\Hat{Q}_t \le (1+\epsilon) Q_t$.
Hence, \eqref{eqn:Qt} holds for all $t$.
This finishes the proof of \Cref{obs:R}.
\end{proof}

    Continue with \eqref{eqn:Pr_s_pt},
    applying \Cref{obs:R} which bounds $\frac{\hat{R}_t}{\hat{Q}_t}$ in terms of $\frac{R_t}{Q_t}$, we get
    \begin{equation}
        \label{eqn:beta_t}
    \beta_t = \frac{1}{K} \cdot \frac{\hat{R}_t}{\hat{Q}_t} \geq \frac{1}{K} \cdot \frac{1-\epsilon}{1+\epsilon} \cdot \frac{R_t}{Q_t}.
    \end{equation}
Next, we establish the following \Cref{lem:cnt},
    to bound $\frac{R_t}{Q_t}$ with respect to $\frac{\sum_{i=1}^n \dist(p_t, p_i)}{Q_n}$.
\begin{lemma}
    \label{lem:cnt}
    For any $t$, if $Q_t > 0$, then the following inequality holds:
\[
    \frac{2R_t}{Q_t} 
    \geq \frac{\sum_{i=1}^n \dist(p_t, p_i)}{Q_n}.
\]
\end{lemma}
\begin{proof}
We first note that it suffices to prove that for every $a \in \{t+1, \dots, n\}$, it holds that 
$\frac{\dist(p_t, p_a)}{2 \sum_{j=1}^t \dist(p_a, p_j)}
\leq \frac{2R_t}{Q_t}$. 
Indeed, if this holds, then
\[
\begin{aligned}
&\frac{\sum_{i=1}^n \dist(p_t, p_i)}{Q_n}\\
=&\frac{R_t + \sum_{i=t+1}^n \dist(p_t, p_i)}{Q_t + 2\sum_{i=t+1}^n \sum_{j=1}^t \dist(p_i, p_j)+\sum_{i=t+1}^n\sum_{j=t+1}^n\dist(p_i,p_j)}\\
\leq& \frac{R_t + \sum_{i=t+1}^n \dist(p_t, p_i)}{Q_t + 2\sum_{i=t+1}^n \sum_{j=1}^t \dist(p_i, p_j)}
\leq \frac{2R_t}{Q_t},
\end{aligned}
\]
where the last inequality follows from the inequality that $\frac{\sum_i c_i}{\sum_i d_i}\leq \max_i \frac{c_i}{d_i}$ for positive numbers $\{c_i\}$ and $\{d_i\}$. This will then complete the proof.

For contradiction, suppose instead that there exists some $a \in \{t+1, \dots, n\}$ with
\begin{equation} \label{equ:sugar_water}
        \frac{\dist(p_t, p_a)}{2 \sum_{j=1}^t \dist(p_a, p_j)} 
        > \frac{2R_t}{Q_t}.
    \end{equation}
We will show that this leads to a contradiction.
Using the definition of $Q_t$, we have:
\[
    Q_t \leq \sum_{i=1}^t \sum_{j=1}^t \left(\dist(p_i, p_t) + \dist(p_j, p_t)\right)
    = 2t \cdot \sum_{i=1}^t \dist(p_i, p_t)
    = 2t\cdot R_t.
\]
Combining this with \Cref{equ:sugar_water}, we can deduce:
$\sum_{j=1}^t \dist(p_a, p_j) 
    < \frac{1}{2} t \cdot \dist(p_t, p_a)$.
Let us analyze $\frac{R_t}{Q_t}$ in this scenario. We have:
\begin{align*}
        \frac{R_t}{Q_t} 
        & \geq \frac{\sum_{i=1}^t \left(\dist(p_t, p_a) - \dist(p_i, p_a)\right)}{\sum_{i=1}^t \sum_{j=1}^t \left(\dist(p_i, p_a) + \dist(p_j, p_a)\right)}  = \frac{t \cdot \dist(p_t, p_a) - \sum_{i=1}^t \dist(p_i, p_a)}{2t \cdot \sum_{i=1}^t \dist(p_i, p_a)} \\
        & > \frac{t \cdot \dist(p_t, p_a) - \frac{1}{2} t \cdot \dist(p_t, p_a)}{2t \cdot \sum_{i=1}^t \dist(p_i, p_a)}  = \frac{\dist(p_t, p_a)}{4 \sum_{i=1}^t \dist(p_i, p_a)}.
    \end{align*}
However, this contradicts \Cref{equ:sugar_water}. Therefore, no such $a$ satisfying \Cref{equ:sugar_water} can exist, and we conclude:
$\frac{\sum_{i=1}^n \dist(p_t, p_i)}{Q_n} 
    \leq \frac{2R_t}{Q_t}$. 
This completes the proof of \Cref{lem:cnt}.
\end{proof}

    Continue with \eqref{eqn:beta_t},
    using Lemma~\ref{lem:cnt} which bounds $\frac{2R_t}{Q_t}$ in terms of $\frac{\sum_{i=1}^n \dist(p_t, p_i)}{Q_n}$, we further obtain:
    \begin{equation}
        \label{eqn:beta_t_final}
    \beta_t \geq \frac{1}{2K} \cdot \frac{1-\epsilon}{1+\epsilon} \cdot \frac{\sum_{i=1}^n \dist(p_t, p_i)}{Q_n}.
    \end{equation}
We still need to bound the sum $\sum_{i=t+1}^n \frac{\hat{R}_i}{\hat{Q}_i}$ in order to bound \eqref{eqn:Pr_s_pt}. Note that
$\frac{\hat{R}_i}{\hat{Q}_i} \le  \frac{1+\epsilon}{1-\epsilon}\cdot\frac{R_i}{Q_i}$, 
and as $Q_k=\sum_{i=1}^k\sum_{j=1}^k\dist(p_i,p_j)=\sum_{i=1}^{k-1}\sum_{j=1}^{k-1}\dist(p_i,p_j)+\sum_{j=1}^k\dist(p_k,p_j)+\sum_{i=1}^{k-1}\dist(p_i,p_k)=Q_{k-1}+2R_k$, the following holds:
\[
    \frac{Q_2}{Q_n}=
    \prod_{k=3}^n \frac{Q_{k-1}}{Q_k} = \prod_{k=3}^n \left(1 - 2 \cdot \frac{R_k}{Q_k}\right).
\]
Using the fact that $(1 - x) \leq e^{-x}$, we derive:
\[
    \prod_{k=3}^n \left(1 - 2 \cdot \frac{R_k}{Q_k}\right) \leq
    e^{-2 \cdot \sum_{k=3}^n \frac{R_k}{Q_k}}.
\]
Since $1 \leq \dist(p_i, p_j) \leq \Delta$ for all $i, j$, it follows that $
    \frac{2}{\Delta n^2} \leq \frac{Q_2}{Q_n}$. 
Combining the above, we have $\frac{2}{\Delta n^2}\leq e^{-2 \cdot \sum_{k=3}^n \frac{R_k}{Q_k}}$ and thus $
    \sum_{k=3}^n \frac{R_k}{Q_k} \leq \ln(\Delta n^2)$. 
Furthermore, 
\begin{equation}
    \label{eqn:hatR_hatQ}
    \sum_{k=2}^n \frac{\hat{R}_k}{\hat{Q}_k} \leq 1 + \frac{1+\epsilon}{1-\epsilon}\cdot\ln(\Delta n^2) \leq \frac{K}{2}.
\end{equation}
Substituting this and \eqref{eqn:beta_t_final} back to \eqref{eqn:Pr_s_pt}, we obtain:
\[
   \Pr[s = p_t] \geq \beta_t \cdot \left(1 - \frac{1}{K} \cdot \sum_{i=t+1}^n \frac{\hat{R}_i}{\hat{Q}_i}\right)
    \geq \frac{1}{2K} \cdot \frac{1-\epsilon}{1+\epsilon} \cdot \frac{\sum_{i=1}^n \dist(p_t, p_i)}{Q_n} \cdot \left(1 - \frac{1}{K} \cdot \frac{K}{2}\right).
\]
This yields 
    $\Pr[s = p_t] \geq \frac{1-\epsilon}{4K(1+\epsilon)} \cdot \frac{\sum_{i=1}^n \dist(p_t, p_i)}{Q_n}$, 
which concludes the $Q_t > 0$ case.

\subparagraph{Case II: $Q_t = 0$.}
We turn to the remaining case $Q_t = 0$.
Let $j\ge t$ be the largest index that $Q_j=0$.
    Then for any $a,b \le j$,
    $\dist(p_a,p_b) = 0$.
    We have
    \begin{align}
        \Pr[s = p_t]  
        & = \beta_t \cdot \prod_{i = t+1}^n (1 - \beta_i) \nonumber  = \frac{1}{t} \cdot \prod_{i = t+1}^ j (1-\frac{1}{i})\cdot \prod_{i = j+1}^n (1 - \beta_i) \nonumber \\
        & \ge \frac 1j \cdot \left(1 - \frac1K \cdot \sum_{i=j+1}^n \frac{\hat{R}_i}{\hat{Q}_i}\right). \label{eqn:Qt_eq_zero}
    \end{align}
Since $Q_i \neq 0$ for $i \geq j + 1$, we can apply \eqref{eqn:hatR_hatQ} to get
    \begin{equation}
        \label{eqn:sum_hatR_hatQ_Qt_eq_zero}
        \sum_{i=j+1}^n \frac{\hat{R}_i}{\hat{Q}_i} \le \frac{K}2.
    \end{equation}
To also lower bound $ 1/ j$, we establish the following \Cref{lem:beta}.
\begin{lemma}
    \label{lem:beta}
    For any $t$, if $Q_t = 0$, then the following inequality holds:
\[
    \frac{1}{t} \ge \frac{\sum_{i=1}^n \dist(p_t, p_i)}{Q_n}.
\]
\end{lemma}
\begin{proof}
    Note that if $Q_t = 0$, then for $i\le t$, we have $\dist(p_i, p_t) = 0$. We get
\[
        Q_n \ge \sum_{i = 1}^t \sum_{j = 1}^n \dist(p_i, p_j) \ge \sum_{i = 1}^t \sum_{j = 1}^n \left(\dist(p_t, p_j) - \dist(p_i, p_t)\right) = t\cdot \sum_{j = 1}^n \dist(p_t, p_j).
\]
Thus, we have that 
$\frac{1}{t} \ge \frac{\sum_{i=1}^n \dist(p_t, p_i)}{Q_n}$. 
This finishes the proof of \Cref{lem:beta}. 
\end{proof}
Finally, Substituting \Cref{lem:beta} and \eqref{eqn:sum_hatR_hatQ_Qt_eq_zero}  
to \eqref{eqn:Qt_eq_zero}, we have
\[
        \Pr[s = p_t] 
        \ge \frac{\sum_{i=1}^n \dist(p_j,p_i)}{Q_n} \cdot \left(1-\frac12\right)
        \ge \frac{1-\epsilon}{4K(1+\epsilon)} \cdot \frac{\sum_{i=1}^n \dist(p_t, p_i)}{Q_n}.
\]
This finishes the proof of the key \Cref{lem:sampling lower bound}.
\end{proof}
 \section{Algorithms for Sliding-Window Streams}
\label{sec:sliding_window}
In this section, we present our sliding-window algorithm that approximately reports \maxcut value in general metrics.

\begin{theorem}    
    \label{thm:sliding_window}
    Suppose $(V, \dist)$ is an underlying metric space.
    There is a streaming algorithm that, given as input
    $\epsilon \in (0, 1/4), \Delta ,w \geq 1$,
    and a point set $P$ presented as a sliding-window stream with window width $w$, 
    such that if $\dist(x, y) > 0$ then $1\le \dist(x,y)\le \Delta$ for all $x, y\in P$,
    reports for each window a $(1 + \epsilon)$-approximation to its \maxcut value
    with probability $\frac23$.
The algorithm uses space $\poly(\frac{\log (w \Delta)}{\epsilon})$.
\end{theorem}
\begin{proof}
The plan is to employ the smooth histogram framework~\cite{braverman2010effective},
which reduces the sliding-window problem to the insertion-only setting, provided that the insertion-only algorithm is \emph{smooth}.
The definition of smoothness is restated in \Cref{def:smooth}.
For two sequences $A, B$, $B \subseteq_r A$ denotes that $B$ is a prefix of $A$.
\begin{definition}[\cite{braverman2010effective}]
\label{def:smooth}
    Let $w \geq 1$ be the window size, and $(V, \dist)$ is an underlying metric such that $\forall x \neq y, 1 \leq \dist(x, y) \leq \Delta$.
    A function $f : 2^V \to \mathbb{R}$ is $(\alpha, \beta)$-smooth if the following holds for any $A \subseteq V$.
    \begin{enumerate}
        \item $f(A) \geq 0$.
        \item $f(A) \geq f(B)$ for $B \subseteq_r A$.
        \item $f(A) \leq \poly(|A|) \cdot \Delta$.
        \item For any $0 < \epsilon < 1$, there exists $\alpha = \alpha(\epsilon, f)$ and $\beta = \beta(\epsilon, f)$ such that
        \begin{itemize}
            \item $0 < \beta \leq \alpha < 1$;
            \item if $B \subseteq_r A$ and $(1 - \beta)f(A) \leq f(B)$, then $(1 - \alpha)f(A \cup C) \leq f(B \cup C)$ for any $C \subseteq V$.
        \end{itemize}
    \end{enumerate}
\end{definition}

The following reduction was established in \cite{braverman2010effective},
which reduces sliding-window algorithm design to the insertion-only case,
provided that the insertion-only algorithm is smooth.
\begin{lemma}[\cite{braverman2010effective}]
    \label{lemma:smooth}
    Let $w \geq 1$ be the window size, and
    let $f$ be an $(\alpha, \beta)$-smooth function.
    If there exists an algorithm $\mathcal{A}$ that maintains a $(1 + \hat{\varepsilon})$-approximation of $f$ on some input $D$ presented as an insertion-only stream using space $g(\hat{\varepsilon})$,
    with failure probability at most $\delta$,
    then there exists an algorithm $\mathcal{A}'$ that maintains a $(1 + \alpha + \hat{\varepsilon})$-approximation of $f$ on each sliding-window 
    with failure probability at most $\frac{\delta\log n}{\beta}$,
    using space
    $O\left(\frac{1}{\beta}(g(\hat{\varepsilon}) + \log (wM)) \log (wM)\right)$.
\end{lemma}

Our main technical lemma in this section is the following smoothness bound for \maxcut.
Let $S$ be a subset of $P$.
Here, we do not directly analyze the smoothness of $f(S) = \maxcut(S)$. Instead, we add a term $|S| \cdot \epsilon / n$ to the \maxcut value,
and this term is crucially needed.
In fact, without this additional term, one can show that $f(S) := \maxcut(S)$ is \emph{not} smooth (as we mention in \Cref{sec:tech_overview}).
The proof of \Cref{lem:f} can be found in \Cref{sec:proof_smoothness}.
\begin{lemma}[Smoothness of \maxcut]
\label{lem:f}
Let $f : V \to \mathbb{R}$ be $f(S) := \maxcut(S) + |S| \cdot \epsilon / n $.
Then $f$ is $(\epsilon,\frac1{64}\epsilon)$-smooth.\end{lemma}

\paragraph{Concluding \Cref{thm:sliding_window}.}
We start with describing the algorithm.
Let $f(S) := \maxcut(S) + |S| \cdot \epsilon / n$.
Note that the smoothness bound in \Cref{lem:f} is only for $f$ and is not for \maxcut.
Hence, we need to devise an insertion-only algorithm $\mathcal{A}$ to approximate $f$.
To this end, we use \Cref{thm:insertion} in a black-box way to obtain a $(1 + \epsilon)$-approximation for \maxcut.
To reduce the failure probability, we independently run the algorithm 
$O(\log \frac{\log n}{\epsilon})$ times and take the median of the resulting values 
as the final $(1+\epsilon)$-approximation. 
This ensures that the overall failure probability is at most 
$\frac{\epsilon}{192 \log n}$.
The additional term $|S| \cdot \epsilon / n$ can be computed exactly in sliding window using $\tilde O(1)$ space.
Combining the two terms yield the desired $\mathcal{A}$.
Clearly, $\mathcal{A}$ gives a $(1+\epsilon)$-approximation of $f$.
From \Cref{lem:f}, $f$ is $(\epsilon,\frac1{64}\epsilon)$-smooth, 
so we can apply \Cref{lemma:smooth} to turn $\mathcal{A}$ into $\mathcal{A}'$ for sliding-window setting, where $\alpha = \epsilon$ and $\beta := \frac1{64}\epsilon$.
Finally, let $t$ be the value returned by $\mathcal{A}'$,
and if $t > 0.5$ we return $t$ and $0$ otherwise.

Next, we show that the approximation ratio is $1+4\epsilon$.
Specifically, let $W$ be the point set of an arbitrary sliding window, and we analyze the ratio for approximating $\maxcut(W)$.
We start with the case $\maxcut(W)>0$, and we argue that in this case the algorithm must return $t$ as the return value instead of $0$.
To see this, since the minimum non-zero distance in the metric is $1$, we get $\maxcut(W) \geq 1$ and hence $f(W) > \maxcut(W)\ge 1$.
Since $t$ is a $(1+2\epsilon)$-approximation of $f(W)$, $t$ is still greater than $0.5$,
and this implies that the return value of the algorithm is $t$.
Now, to further bound the ratio of $t$,
note that $|W| \cdot \epsilon / n \leq \epsilon$ and recall that we argued $\maxcut(W) \geq 1$,
we know that $f(W) = \maxcut(W) + |W|\cdot \epsilon / n$ is $(1 + \epsilon)$-approximation to $\maxcut(W)$.
Since the return value $t$ is a $(1 + 2\epsilon)$-approximation to $f(S)$ by \Cref{lemma:smooth}, we conclude that the return value $t$ is $(1+4\epsilon)$-approximation of $\maxcut(W)$.

For the remaining case of $\maxcut(W)=0$, we again use that $f(W)= |W|\cdot \epsilon / n\le \epsilon$.
This implies that $(1+2\epsilon)$-approximation of $f(W)$ is less than $0.5$,
therefore our algorithm correctly returns $0$.

By \Cref{lemma:smooth},
the failure probability is at most $\frac{\epsilon}{192 \log n}\cdot \frac{\log n}{\beta}\le \frac13$.
\end{proof}

\subsection{Proof of \Cref{lem:f}: The Smoothness of \maxcut}
\label{sec:proof_smoothness}

Before we proceed, we need the following mathematical facts about \maxcut function.
\begin{fact}
\label{fact:mc}
$\frac{1}{4} \cut(S, S) \leq \maxcut(S) \leq \frac{1}{2} \cut(S, S)$.
\end{fact}
\begin{proof}
Let $E$ be the sum of all pairwise distances of $S$, then it is well known that $\frac12 E \leq \maxcut(S) \leq E$.
Here, we have another factor of $\frac12$ because the distance $\dist(x, y)$ between every two points $x, y$ is considered twice in $\cut(S, S)$.
\end{proof}

\begin{fact}
    For $B \subseteq A \subseteq V$, we have 
    \begin{align*}
        \maxcut(B)+\frac14\cdot\cut(A\setminus B, A\setminus B) + \frac12\cdot \cut(A\setminus B, B)\le \maxcut(A).
    \end{align*}
\end{fact}
\begin{proof}
    Let $B^\prime\subseteq B$ be the set such that $\sum_{x\in B^\prime}\sum_{y\in B\setminus B^\prime}\dist(x,y) = \maxcut(B)$.
    Let $S$ be a random point set defined as follows.
    Initially, $S = B'$.
    Then, for every point in $A\setminus B$,
    we independently put it into $S$ with probability $\frac{1}{2}$.
Now $S\cap B = B'$ and $(A \setminus S)\cap B = B \setminus B'$.
\[
\begin{aligned}
    \mathbb{E}[\cut(S,A \setminus S)] 
        &= \mathbb{E}[\cut(S\cap B, (A \setminus S)\cap B)] 
         + \mathbb{E}[\cut(S\setminus B, (A \setminus S)\cap B)]\\
        &\quad + \mathbb{E}[\cut(S\cap B, (A \setminus S)\setminus B)] 
         + \mathbb{E}[\cut(S\setminus B, (A \setminus S)\setminus B)] \\
        &= \cut(B', B\setminus B') 
         + \tfrac12 \cdot \cut(A\setminus B, B\setminus B') \\
        &\quad + \tfrac12 \cdot \cut(B', A\setminus B) 
         + \tfrac14 \cdot \cut(A\setminus B, A\setminus B)
\end{aligned}
\]

By the optimality of $\maxcut(A)$, we have $\maxcut(A) \geq \E[\cut(S, A \setminus S)]$.
The proof is finished by combining this with the above upper bound of $\E[\cut(S, A\setminus S)]$.
\end{proof}

Let $\alpha := \epsilon, \beta := \frac1{64}\epsilon$. 
To verify the smoothness of $f$ (\Cref{def:smooth}), 
we need to show that if $B \subseteq_r A$ and $(1-\beta)f(A) \leq f(B)$, then $(1-\alpha)f(A \cup C) \leq f(B \cup C)$ for all $C \subseteq P \setminus A$.

If $B = \emptyset$, then $f(A) = f(B) = 0$, so $A = B =\emptyset$.
We can directly get $f(A \cup C) = f(B \cup C)$.
In the following, we assume that $B\ne \emptyset$.
To begin with, we assume that $\cut(B, B) > 0$.
By the premises of $(1-\beta)f(A)\le f(B)$, we have
\[(1-\beta)\left(f(B)+|A\setminus B|\cdot \epsilon / n + \frac14\cut(A\setminus B, A\setminus B)+\frac12\cut(A\setminus B, B)\right)\le f(B)\]
Note that $1-\beta \ge \frac12$, we have
$\frac18\cut(A\setminus B, A\setminus B)+\frac14\cut(A\setminus B, B)\le \beta f(B)$.  
As $|B| \cdot \epsilon / n \le n \cdot \epsilon / n \le 1\le \maxcut(B)$,
we get 
\[
\frac12\cut(A\setminus B, A\setminus B)+\cut(A\setminus B, B)\le 4\beta f(B) = 4\beta\left(\maxcut(B) + \epsilon / n\cdot |B|\right)\le 4\beta \cut(B, B)
\]
This further implies
\begin{equation}\label{eq:part2}
\cut(A\setminus B, A\setminus B)\le 8\beta \cut(B, B)
\end{equation}
and 
\begin{equation}\label{eq:part1}
\cut(A\setminus B, B)\le 4\beta \cut(B, B)
\end{equation}

Now, for every $u\in A\setminus B$, we have
\begin{equation*}
\cut(B, B)\le \sum_{x\in B}\sum_{y\in B}(\dist(x,u)+\dist(u,y))=2|B|\cdot \sum_{y\in B}\dist(u,y)
\end{equation*}
Summing all such $u \in A \setminus B$, we get $ \cut(A\setminus B, B)\ge \frac{|A\setminus B|}{2|B|}\cut(B, B)$. 
Combining this inequality with \Cref{eq:part1},
we have
\begin{equation}
\label{eqn:frac_AmB_B}
 \frac{|A\setminus B|}{|B|}\le 8\beta
\end{equation}
Consider the case that $\cut(B, B) = 0$.
Assume that $\cut(A, A) > 0$.
Then $(1-\beta)f(A)\ge (1-\beta) > n\cdot \epsilon / n \ge f(B) = |B| \cdot \epsilon / n$.
This is impossible, so $\cut(A, A) = 0$.
From $(1-\beta)|A|\le |B|$,
we can still obtain \Cref{eq:part2}, \Cref{eq:part1} and \Cref{eqn:frac_AmB_B}.

Next, fix some $c \in C$. By triangle inequality,
$
0\le \sum_{x\in A\setminus B}\sum_{y\in B}[\dist(x,y)+\dist(y,c)-\dist(x,c)]
$
so
$
|B|\cdot \sum_{x\in A\setminus B}\dist(x,c)\le \cut(A\setminus B, B)+|A\setminus B|\sum_{y\in B}\dist(y,c). 
$
Combining this inequality with \Cref{eq:part1}, we get
\begin{equation}
\label{eqn:B_AmB}
    |B|\cdot \sum_{x\in A\setminus B}\dist(x,c)\le 4\beta \cut(B, B)+|A\setminus B|\sum_{y\in B}\dist(y,c)
\end{equation}
On the other hand,
\begin{equation}
\label{eqn:cut_BB}
    \cut(B, B)\le \sum_{x\in B}\sum_{y\in B}(\dist(x,c)+\dist(y,c))=2|B|\cdot \sum_{x\in B}\dist(x,c)
\end{equation}
Combining \eqref{eqn:B_AmB} and \eqref{eqn:cut_BB}, it holds that
\[
|B|\cdot \sum_{x\in A\setminus B}\dist(x,c)\le 8\beta|B|\sum_{x\in B}\dist(x,c)+|A\setminus B|\sum_{x\in B}\dist(x,c),
\]
which implies that
\[
\sum_{x\in A\setminus B}\dist(x,c)\le 8\beta \sum_{x\in B}\dist(x,c)+\frac{|A\setminus B|}{|B|}\sum_{x\in B}\dist(x,c)\le 16\beta \sum_{x\in B}\dist(x,c), 
\]
where the second inequality is from \eqref{eqn:frac_AmB_B}.
Summing over $c \in C$, we obtain 
$
\cut(A\setminus B, C)\le 16\beta\cut(B, C).
$ 
This together with \eqref{eq:part2} and \eqref{eq:part1} implies
\[
\begin{aligned}
\cut(A\setminus B, A\cup C)
&\leq \cut(A \setminus B, A \setminus B) + \cut(A \setminus B, B) + \cut(A \setminus B, C) \\
&\leq 12\beta \cut(B, B) + 16\beta \cut(B, C) \le 16\beta\cut(B\cup C,B\cup C) \\
&= \frac14\alpha\cut(B\cup C,B\cup C)  \leq \alpha \maxcut(B\cup C),
\end{aligned}
\]
where the last inequality follows from \eqref{fact:mc}.
This implies
\[
    (1-\alpha)[\cut(A\setminus B, A\cup C)+\maxcut(B\cup C)]\le \maxcut(B\cup C)
\]
From the definition of \maxcut, 
there exists $S\subseteq A\cup C$,
such that $\maxcut(A\cup C) = \cut(S, (A\cup C)\setminus S)$.
Write $S' := (A\cup C)\setminus S$.
We have
\[
\begin{aligned}
\maxcut(A \cup C)
    &= \cut(S, S') \\
    &= \cut(S\cap (B \cup C), S'\cap (B \cup C)) + \cut(S\cap (B \cup C), S'\cap (A\setminus B)) \\
    &\quad + \cut(S\cap (A\setminus B), S') \\
    & \le  \maxcut(B\cup C) + \cut(A\cup C, S'\cap (A\setminus B)) + \cut(S\cap (A\setminus B), A\cup C) \\
    & \le \maxcut(B\cup C) + \cut(A\setminus B, A\cup C)\\
    & \le (1+\alpha)\maxcut(B\cup C)
\end{aligned}
\]

Note that $
    (1-\alpha)\cdot |A|\le (1-8\beta)\cdot (1+8\beta)|B|\le |B|$. 
Hence, 
$
    (1-\alpha)f(A\cup C)\le f(B\cup C). 
$
This verifies \Cref{def:smooth} and finishes the proof of \Cref{lem:f}.\qed

\section{Lower Bound for Dynamic Streams}\label{app:lowerbound}
In this section, we establish the lower bound for the problem in the dynamic streaming setting.
We use the same model as in our upper bounds, specifically the one defined in~\Cref{def:model}.

\begin{theorem}
    \label{thm:lb}
    For any streaming algorithm $\mathcal{A}$
    and sufficiently large integer $n$,
    there exists a metric space and dynamic point (ID) stream, both of size $O(n)$, 
    such that
$\mathcal{A}$ must use $\Omega(n^{1 / 3})$ space to compute $O(\Delta / n^{1 / 3})$-approximation of $\maxcut$ with probability at least $0.55$ 
    where $\Delta$ is the aspect ratio of the dataset.
\end{theorem}
To prove \Cref{thm:lb}, we will first establish the following lemma.

\begin{lemma}
    \label{lem:lb}
    There exists an input distribution $\mathcal{D}$, over the metric space and a dynamic point (ID) stream in this metric, both of size $O(n)$,
    such that for any deterministic algorithm $\mathcal{A}$,
    when its input is sampled from $\mathcal{D}$,
$\mathcal{A}$ must use $\Omega(n^{1 / 3})$ space to compute $O(\Delta / n^{1 / 3})$-approximation of $\maxcut$ with probability at least $0.55$.
\end{lemma}

We first show that \Cref{lem:lb} directly implies our lower bound result, namely \Cref{thm:lb}. The transition from \Cref{lem:lb} to \Cref{thm:lb} follows from Yao’s minimax principle~\cite{DBLP:conf/focs/Yao77}. For completeness, we restate the argument below.

\begin{proof}[Proof of \Cref{thm:lb}]
For any randomized algorithm $\mathcal{A}$, let $\sigma$ denote the distribution of its random seed.
Define $\mathcal{E}$ as the event that, on input $T$ sampled from $\mathcal{D}$ and seed $r \sim \sigma$, 
algorithm $\mathcal{A}$ outputs an $O(\Delta / n^{1 / 3})$-approximation using $o(n^{1/3})$ bits of space. 
We can write
\[
    \mathbb{P}_{T \sim \mathcal{D},\, r \sim \sigma}[\mathcal{E}]
    = \mathbb{E}_{r \sim \sigma}\Bigl[\mathbb{P}_{T \sim \mathcal{D}}[\mathcal{E} \mid r]\Bigr].
\]
Fixing the random seed $r$ turns $\mathcal{A}$ into a deterministic algorithm. 
By \Cref{lem:lb}, for any fixed $r$, we have 
\(\mathbb{P}_{T \sim \mathcal{D}}[\mathcal{E} \mid r] \le 0.55\).
Thus,
\[
    \mathbb{P}_{T \sim \mathcal{D},\, r \sim \sigma}[\mathcal{E}]
    \le 0.55.
\]
Then there exists a $T$, such that:
\[
    \mathbb{P}_{r \sim \sigma}[\mathcal{E}\mid T]
    \le 0.55,
\]
which completes the proof.
\end{proof}

\subsection{Proof of \Cref{lem:lb}}
\label{sec:proof_lem_lb}
In this section, we give the proof of \Cref{lem:lb}. We start with defining the (hard) input for the lower bound.

\paragraph{The Basic Structures of Points and the Metric.} Let $P := \{p_{i,j} \mid i \in [n^{\frac{2}{3}}], j \in [n^{\frac{1}{3}}]\}$ be the point set.
The distance function $\dist : P \times P \to \mathbb{R}_+$ is defined with respect to additional (integer) parameters
$\Delta \geq \Omega(n^{1 / 3}) \geq 1$,
$K \in \{1, \Delta\}$ and $i^*\in[n^{2/3}]$ and $j^*\in [n^{1/3}]$ (whose values are random and will be picked later).
The distance function over $P$ is defined as follows:
\begin{equation}
\label{eqn:def_dist}
    \dist(p_{i,j}, p_{k,l}) =
    \begin{cases}
        \Delta, & i \neq k, \\
        1, & i = k \text{ and } p_{i^*,j^*} \notin \{p_{i,j}, p_{k,l}\}, \\
        K, & i = k \text{ and } p_{i^*,j^*} \in \{p_{i,j}, p_{k,l}\},
    \end{cases}
\end{equation}
where the trivial case $p_{i,j} = p_{k,l}$ is omitted. 
To show that the distances satisfy the triangle inequality,
we assume that on the contrary that there exist distinct points $p_{i_1,j_1},p_{i_2,j_2},p_{i_3,j_3}$, such that
$$\dist(p_{i_1,j_1},p_{i_2,j_2}) + \dist(p_{i_2,j_2},p_{i_3,j_3}) < \dist(p_{i_1,j_1},p_{i_3,j_3}).$$
Note that for $x\ne y$, $\dist(x,y)\in\{1,\Delta\}$. It must be:
\begin{align}
    \dist(p_{i_1,j_1},p_{i_2,j_2}) &= \dist(p_{i_2,j_2},p_{i_3,j_3}) = 1 
    \label{eq:1} \\
    \dist(p_{i_1,j_1},p_{i_3,j_3}) &= \Delta
    \label{eq:2}
\end{align}
From \Cref{eq:1}, we have $i_1 = i_2 = i_3$. 
As $\dist(p_{i_1,j_1},p_{i_3,j_3}) = \Delta$, \Cref{eq:1} forces $K = \Delta$, with one of $j_1,j_3$ equal to $j^*$.
Then one of $\dist(p_{i_1,j_1},p_{i_2,j_2}), \dist(p_{i_2,j_2},p_{i_3,j_3})$ is $\Delta$,
which contradicts with \Cref{eq:1}. Thus, the triangle inequality holds.

\paragraph{Distribution of the Metric and Points.} We next define a distribution over potential identifiers (IDs) for the points, as well as the special indices $i^*,j^*$. 

\begin{itemize}
    \item For each point $p_{i,j}$, define an identifier $\text{ID}_{i,j}:=(i, j, a_{i,j})$, 
where each $a_{i,j}$ is drawn independently and uniformly at random from $\{0, 1\}$. Let $U:=\{(i,j,a_{i,j})\mid i \in [n^{\frac{2}{3}}], j \in [n^{\frac{1}{3}}]\}$ be the set of all IDs.
Note that there is a one-to-one correspondence between $U$ and $P$.
\item The index $i^*$ is sampled uniformly from $[n^{\frac{2}{3}}]$, 
and $j^*$ uniformly from $[n^{\frac{1}{3}}]$. 
\end{itemize}

For $K\in\{1,\Delta\}$, let ${\cal D}^K$ denote the above distribution over the points, IDs and indices $i^*,j^*$ with the corresponding $K$-value in the distance function. We define the distribution $\cal D$ as follows: (1) pick $K=1$ with probability $1/2$ and $K=\Delta$ otherwise; (2) draw IDs and indices $i^*,j^*$ as above from ${\cal D}^K$. 
Thus, an instance $T\sim \mathcal{D}$ can be written as the 4-tuple $(U, i^*, j^*, K)$.

\paragraph{The Data Stream.}

For convenience, define $C_t := \{p_{j,t} \mid j \in [n^{\frac{2}{3}}]\}$ for $t \in [n^\frac13]$ and $R_t := \{p_{t,i} \mid i \in [n^{\frac{1}{3}}]\}$ for $t \in [n^\frac23]$. 
Intuitively, the sets $C_t$ correspond to the $t$-th column of this matrix-like arrangement of points, while the sets $R_t$ correspond to the $t$-th row.

We describe the stream as follows, where we use the language of points in $P$
but in fact they are presented to the algorithm as the corresponding point in $U$ (recalling that there is a one-to-one correspondence between $U$ and $P$):
\begin{enumerate}
    \item Insert points in $C_1$ in lexicographical order of the row index: $p_{1,1},p_{2,1},\ldots,p_{n^{\frac{2}{3}},1}$ 
    \item Continue column by column, inserting $C_2, C_3, \ldots, C_{n^{\frac{1}{3}}}$,
    each in lexicographical order (of the row index).
    \item Delete all points in $P \setminus R_{i^*}$, in an arbitrary order.
\end{enumerate}

Note that the remaining points are exactly those in $R_{i^*}$.

We will prove that no (deterministic) algorithm using less than $0.01 \cdot n^{\frac{1}{3}}$ bits of space can compute an estimate $\eta > 0$ that is, 
with high probability, an $O(\Delta / n^{1 / 3})$-approximation to the Max-Cut value of the remaining points. 
Let $m := 0.01 \cdot n^{\frac{1}{3}}$.

To this end, we make the following crucial observation such that the \maxcut value differs by a $\poly(n)$ factor between $K = 1$ and $K = \Delta$ cases.
Hence, it suffices to show that the algorithm cannot differentiate if $K = 1$ or $K = \Delta$ in small space.

\begin{fact}
    \label{fact:K-gap}
When $K=1$, the Max-Cut value is at most $n^{2/3}/2$; when $K=\Delta$, the Max-Cut value is at least $\Delta(n^{1/3}-1)$.
\end{fact}
\begin{proof}
If $K=1$, all remaining points are within the same row $R_{i^*}$, with pairwise distances equal to $1$. The maximum cut in a clique of size $n^{1/3}$ is $\lfloor \frac{n^{\frac{1}{3}}}{2} \rfloor \cdot (n^{\frac{1}{3}} - \lfloor \frac{n^{\frac{1}{3}}}{2} \rfloor)\leq n^{2/3}/2$.

If $K=\Delta$, then the point $p_{i^*,j^*}$ is an outlier: its distance to any other point is $\Delta$, while all other pairs are at distance $1$.
Thus, the maximum cut separates $p_{i^*,j^*}$ from the rest, achieving cut value $\Delta \cdot (n^{\frac{1}{3}} - 1)$.
\end{proof}

\begin{proof}[Proof of \Cref{lem:lb}]
By \Cref{fact:K-gap}, it suffices to show that the algorithm cannot differentiate $K = 1$ or $K = \Delta$.
Recall that in \eqref{eqn:def_dist}, the algorithm learns $K$ only if the algorithm queries the distance between $p_{i^*, j^*}$ and another point in $R_{i^*}$,
since otherwise, the input to the algorithm would be the same, and the algorithm must give the same return value because the algorithm is deterministic,
in which case the algorithm cannot differentiate the value of $K$.
This is a random event because of the randomness of $i^*$ and $j^*$, and we define $\mathcal{E}$ as this crucial event, as follows.
\begin{definition}
Let $\cal E$ be the event that the algorithm has ever queried the distance between $p_{i^*,j^*}$ and another point in $R_{i^*}$.
\end{definition}
The main technical lemma is to show $\mathcal{E}$ happens with small probability, 
which is stated as follows whose proof is postponed to \Cref{sec:proof_lemma_bound_E},
and this would conclude \Cref{lem:lb}.
\begin{lemma}
\label{lemma:bound_E}
    $\Pr[\mathcal{E}] \leq 0.1$.
\end{lemma}
In conclusion, by \Cref{lemma:bound_E} the algorithm cannot differentiate whether $K = 1$ or $K = \Delta$ with probability $0.9$.
By \Cref{fact:K-gap}, which readily implies that the algorithm must suffer a $\poly(n)$ factor error,
provided that $\Delta$ is set to high-degree polynomial of $n$.
\end{proof}

\subsubsection{Proof of \Cref{lemma:bound_E}: Bounding $\Pr[\mathcal{E}]$}
\label{sec:proof_lemma_bound_E}

Recall that the input stream can be broken into two parts: a first part that only consists of insertions,
and a second part that only consists of deletions.
We call the first the \emph{insertion stage}, and the second part the \emph{deletion stage}.

Observe that $\dist$ is random and it particularly depends on $i^*$ and $j^*$.
Consequently, this potentially makes the internal state of the deterministic algorithm depending on $i^*$ and $j^*$,
and this makes it difficult to analyze the event $\mathcal{E}$.
We will see that this does not matter and there is a clean way around it.
Intuitively, consider the earliest time in the stream that $\mathcal{E}$ happens,
then before this time step, the algorithm is oblivious to the value of $K$,
and by \eqref{eqn:def_dist}, $\dist = K$ if and only if the distance between $p_{i^*,j^*}$ and some other point in $R_{i^*}$ is queried,
one can conclude that up to this point, the algorithm is also oblivious to $i^*, j^*$.

To make this argument formal, consider the following auxiliary distance function $\dist'$
\begin{equation}
\label{eqn:def_distp}
    \dist'(p_{i,j}, p_{k,l}) =
    \begin{cases}
        \Delta, & i \neq k, \\
        1, & i = k \\
    \end{cases}
\end{equation}
Let $\X$ be the event that the algorithm queries $\dist'$ 
with two distinct points in $R_{i^*}$ during the insertion stage.
Let $\Y$ be the event that the algorithm queries $\dist'$ between 
$p_{i^*,j^*}$ and another point in $R_{i^*}$ during the deletion stage.
Next, we show that the event $\mathcal{E}$ is upper bounded by the union of event $\X$ and $\Y$.
\begin{lemma}
\label{lemma:calE_XY}
    $\Pr[\mathcal{E}] \leq \Pr[\X \vee \Y]$.
\end{lemma}
\begin{proof}
    Define two intermediate events $\X', \Y'$,
    such that $\X'$ is the event that the algorithm queries $\dist$ 
    with two distinct points in $R_{i^*}$ during the insertion stage,
and $\Y'$ is that the algorithm queries $\dist$ between 
    $p_{i^*,j^*}$ and another point in $R_{i^*}$ during the deletion stage.
    It is immediate that $\mathcal{E}$ implies $\X' \vee \Y'$. Indeed, if $\mathcal{E}$ happens,
    then a distance query is made to a pair containing $p_{i^*, j^*} \in R_{i^*}$ either in the insertion stage which implies $\X'$, or the deletion stage which implies $\Y'$.
    Therefore, $\Pr[\mathcal{E}] \leq \Pr[\X' \vee \Y']$.
    
    Now, we argue that $\Pr[X' \vee \Y'] = \Pr[\X \vee \Y]$, and this would conclude the lemma.
    To this end, we first argue that $X'$ happens if and only if $X$ happens.
    A key observation is that $\dist$ and $\dist'$ agrees on all point pairs that do not contain $p_{i^*,j^*}$ (recalling \eqref{eqn:def_dist} and \eqref{eqn:def_distp}).
    Now, fix a realization of $i^*$ and $j^*$.
    As long as the algorithm does not query a pair that contains $p_{i^*, j^*}$ (in the insertion stage),
    the response of distance remain the same regardless of $\dist$ or $\dist'$,
    and this further means the query sequence of the algorithm is the same since the algorithm is deterministic.
    This implies that the time step of the insertion stage such that $X$ or $X'$ happen is the same, denoted as $T$.
    Therefore, if both $X$ and $X'$ do not happen at $T$, then they must both happen immediately after $T$;
    moreover, if $X$ ($X'$ resp.) happens on or before $T$, then $X'$ ($X$ resp.) must also happen because the algorithm is deterministic and the input is the same.
    We conclude that $\Pr[X] = \E[\Pr[X \mid i^*, j^*]] = \E[\Pr[X' \mid i^*, j^*]] = \Pr[X']$.
Similarly, still condition on the realization of $i^*, j^*$,
    if neither $X$ nor $X'$ happens, then either $Y$ and $Y'$ both happens or neither of the two happens. 
    Hence, this concludes that $\Pr[\X' \vee \Y'] = \Pr[\X \vee \Y]$

    This finishes the proof of \Cref{lemma:calE_XY}.
\end{proof}
By union bound, $\Pr[\X \vee \Y] \leq \Pr[\X] + \Pr[\Y]$,
and hence it suffices to bound $\Pr[\X]$ and $\Pr[\Y]$ respectively.

\paragraph{Insertion Stage: Bounding $\Pr[\X]$.}
\label{subsub:insertion}

During the insertion stage, the columns $C_1, \ldots, C_{n^{1/3}}$ arrive sequentially.  
For each $t \in [n^{1/3}]$, let $\X_t$ denote the event that $\X$ occurs during the insertion of $C_t$, i.e., from the time the algorithm reads $\text{ID}_{t,1}$ until it is ready to read $\text{ID}_{t+1,1}$  
(for $t = n^{1/3}$, this stage ends when the first deletion begins).  
By the union bound,
$\Pr[\X] \le \sum_{t=1}^{n^{1/3}} \Pr[\X_t]$.
It therefore suffices to bound $\Pr[\X_t]$ for each $t$.

Fix a column $t$.  
Let $V \subseteq U$ be the set of IDs of $C_t$.
Let $I$ be the internal memory state of the algorithm immediately before the insertion of $C_t$, represented as an $m$-bit string.  
Denote this query sequence of point pairs as $S$,
and let $T := \{ s : (s, t) \in S \} \cup \{ t : (s, t) \in S \}$
be the set of \emph{distinct} point IDs that belong to any query pair (so $T$ is \emph{not} a multiset).
Furthermore, exclude the IDs in $C_t$ from $T$, i.e., let $T' := T \setminus V$.

\begin{fact}
    \label{fact:Rprime}
    $T'$ is determined by $V$ and $I$, and is independent of $i^*$.
\end{fact}
\begin{proof}
The fact that $T'$ is determined by $I$ and $V$ is by definition.
The randomness of $V$ comes from $U$ which is the IDs.
On the other hand, since the algorithm is deterministic,
$I$ depends on $\dist'$ (which is deterministic) as well as the set of IDs ever inserted (whose randomness comes from $U$, and more specifically from the IDs of $C_1, \ldots, C_{t - 1}$).
\end{proof}

\begin{lemma}
    \label{lemma:PrRprime}
    With respect to the randomness of IDs, $\Pr[|T'| > 2m]  \leq \frac{1}{2^m}$.
\end{lemma}
\begin{proof}

    The fact that $T'$ depends on $I$ which depends on $U$ (the IDs)
    makes the analysis of randomness tricky,
    and here we need to do a low-level counting analysis with respect to the behavior of the deterministic algorithm.

    This lemma is trivial when $t = 1$, since $|T'|$ can only be $0$ as we remove the IDs of $C_1$ which is the only input points seen so far.
    Next, we focus on $t \geq 2$.

    We analyze the conditional probability $\Pr[|T'| > 2m \mid V]$, for any $V$,
    and we fix some $V$ in the following analysis.
    For $Z \in \{0, 1\}^m$,
    define $T'_{Z}$ be the realization of $T'$ such that $I = Z$ (and given $V$).
    Since the algorithm is deterministic, we can interpret the behavior of the algorithm as two mappings:
    a) a mapping $f$ that maps the input IDs on $C_1, \ldots, C_{t - 1}$, denoted as $U_{t-1}$, to the internal state $I$, i.e., $f : U_{t - 1} \mapsto I$;
    and b) a mapping $g$ that maps the internal state $I$ to  $T'_{I}$,
    i.e., $g : I \mapsto T'_{I}$.
    Hence,
    \begin{align}
        \Pr[|T'| > 2m \mid V]
        = \sum_{Z \in \{0, 1\}^m}  \mathbb{I}[|T'_{Z}| > 2m] \cdot
        \frac{|f^{-1}(Z)|}{2^{|U_{t - 1}|}},
        \label{eqn:PrRPrime}
    \end{align}
    where $\mathbb{I}$ is the indicator function.
    If $|T'_{Z}| > 2m$, we take the $2m$-length prefix of $T'_{Z}$,
    denoted as $T''_Z$,
    and this $T''_Z \subseteq U_{t - 1}$ defines a specific realization of IDs $(i, j, a_{i, j})$ of some $2m$ positions $(i, j)$'s in $C_1, \ldots, C_{t - 1}$.
    Note that for different $Z$ this $T''_Z$ may also be different.

    Now, denote the subset of $f^{-1}(Z)$ that satisfies a realization $T''_Z$ as $F_{Z, T''_Z}$.
    Then $\frac{|F_{Z, T''_Z}|}{2^{|U_{t-1}|}}$ is the probability that we see this specific realization.
    Therefore, continue from \eqref{eqn:PrRPrime}
    \begin{align}
        \Pr[|T'| > 2m \mid V]
        = \sum_{Z \in \{0, 1\}^m}  \mathbb{I}[|T'_{Z}| > 2m] \cdot 
        \frac{|f^{-1}(Z)|}{2^{|U_{t - 1}|}} 
        \leq \sum_{Z \in \{0, 1\}^m} \frac{|F_{Z, T''_Z}|}{2^{|U_{t - 1}|}}.
        \label{eqn:PrRPrimeSum}
\end{align}
    Observe that $F_{Z, T''_Z}$ is a subset of the $2^{|U_{t-1}|}$ realizations of IDs in $U_{t-1}$,
    such that a fixed $2m$ locations of $(i, j) \in C_1 \cup \ldots \cup C_{t-1}$'s takes a fixed realization in $\{0, 1\}^{2m}$, as specified in $T''_Z$.
    Therefore, at most $|U_{t-1}| - 2m$ locations are free from the constraint
    and $|F_{Z, T''_Z}| \leq 2^{|U_{t - 1}| - 2m}$.
    Substituting this to \eqref{eqn:PrRPrimeSum}, we have
    \begin{align*}
        \Pr[|T'| > 2m \mid V]
\leq \sum_{Z \in \{0, 1\}^m} \frac{|F_{Z, T''_Z}|}{2^{|U_{t - 1}|}} 
        \leq \sum_{Z \in \{0, 1\}^m} \frac{1}{2^{2m}} 
        = \frac{1}{2^m}.
\end{align*}
    We finish the proof of \Cref{lemma:PrRprime} by taking
        $\Pr[|T'| > 2m ] = \E[\Pr[|T'| > 2m \mid V]] \leq \frac{1}{2^m}$.
\end{proof}

Now, condition on the event that $|T'| \leq 2m$,
since $T'$ is independent of $i^*$ (\Cref{fact:Rprime}),
we conclude that
$\Pr[\X_t \mid |T'| \leq 2m] \leq \frac{2m}{n^{2 / 3}}$,
where we are using that $i^*$ is chosen uniformly at random from $[n^{2 / 3}]$
and $T'$ has $2m$ places to hit $i^*$.
Hence, by \Cref{lemma:PrRprime}
\begin{equation*}
    \Pr[\X_t] \leq \Pr[\X_t \mid |T'| \leq 2m]  + \Pr[|T'| > 2m]
\leq \frac{1}{2^m} + \frac{2m}{n^{2 / 3}}.
\end{equation*}
This further implies that
\begin{equation}
\Pr[\X] \leq \sum_{t= 1}^{n^{1 / 3}} \Pr[\X_t] \leq n^{1 / 3} \left( \frac{1}{2^m} + \frac{2m}{n^{2 / 3}}  \right) \leq 0.05.
\end{equation}

\paragraph{Deletion Stage: Bounding $\Pr[\Y]$.}
    Recall that $\Y$ is the event that the algorithm queries $\dist'$
    between $p_{i^*,j^*}$ and another point in $R_{i^*}$ during the deletion stage.
    The proof is similar to that for insertion-only,
    and we only sketch the key steps and the differences.

Condition on $i^*$, and we use the randomness of $j^*$ and the IDs,
    i.e., we analyze $\Pr[\Y \mid i^*]$.
    Let $G \subseteq U$ be the set of IDs of $P \setminus R_{i^*}$,
and let $I$ be the internal memory state of the algorithm before the deletion, represented as an $m$-bit string.  

    Similar to the insertion stage, for each realization $Z \in \{0, 1\}^m$
    of the internal states of the algorithm,
    we define $T$ as the distinct point IDs that belong to any query pair,
    and let $T' := T \setminus G$ be that excluding IDs in $G$.
We can establish a similar lemma as in \Cref{lemma:PrRprime}:
    we have $\Pr[ |T'| > 2m \mid i^* ] \leq \frac{1}{2^m}$,
    which follows from the randomness of the IDs.
    
    Next, observe that $T'$ is independent of $j^*$,
    and we have $\Pr[\Y \mid |T'| \leq 2m, i^*]
    \leq \frac{2m}{n^{1/3}}$,
    where we are using the randomness of $j^*$,
    which is chosen uniformly at random from $[n^{1 / 3}]$ and $T'$ has $2m$ places to hit $j^*$. Hence,
    \begin{align*}
        \Pr[\Y \mid i^*]
        \leq \Pr[\Y \mid |T'| \leq 2m, i^*] + \Pr[|T'| > 2m \mid i^*]
        \leq \frac{1}{2^m} + \frac{2m}{n^{2 / 3}} \leq 0.05.
    \end{align*}
    Therefore, $\Pr[\Y] \leq \E[\Pr[\Y \mid i^*]] \leq 0.05$.

\begin{proof}[Proof of \Cref{lemma:bound_E}]
    In conclusion, we have shown $\Pr[\mathcal{E}] \leq \Pr[\X \vee \Y] \leq \Pr[\X] + \Pr[\Y]$,
    and combining with the upper bounds for $\Pr[\X]$ and $\Pr[\Y]$ we just established,
    we obtain the desired upper bound of $\Pr[\mathcal{E}]$, and this finishes the proof.
\end{proof}

 \bibliography{ref.bib}

\end{document}